\documentclass[a4paper, 12pt]{article}

\usepackage[utf8]{inputenc}
\usepackage[T1]{fontenc}
\usepackage{fullpage}

\usepackage{amsmath, amsthm, amssymb}
\usepackage{graphicx}
\usepackage{enumerate}
\usepackage{enumitem}
\usepackage{authblk}

\usepackage[noend]{algpseudocode}
\usepackage{algorithm}
\usepackage{thmtools}
\usepackage{thm-restate}
\usepackage[colorlinks=true, citecolor=red]{hyperref}

\usepackage{tikz}
\renewenvironment{abstract}
{\small\vspace{-1em}
\begin{center}
\bfseries\abstractname\vspace{-.5em}\vspace{0pt}
\end{center}
\list{}{
\setlength{\leftmargin}{0.6in}%
\setlength{\rightmargin}{\leftmargin}}%
\item\relax}
{\endlist}
\usepackage{comment}
\declaretheorem[name=Theorem, numberwithin=section]{theorem}
\declaretheorem[name=Lemma, sibling=theorem]{lemma}

\declaretheorem[name=Definition, sibling=theorem]{definition}

\declaretheorem[name=Claim, sibling=theorem]{claim}

\declaretheorem[name=Remark, style=remark, sibling=theorem]{remark}

\declaretheorem[name=Question, sibling=theorem]{question}
\def\cqedsymbol{\ifmmode$\lrcorner$\else{\unskip\nobreak\hfil
\penalty50\hskip1em\null\nobreak\hfil$\lrcorner$
\parfillskip=0pt\finalhyphendemerits=0\endgraf}\fi}

\newcommand{\cG}{\mathcal{G}}

\DeclareMathOperator{\diam}{diam}

\let\le\leqslant
\let\ge\geqslant
\let\leq\leqslant
\let\geq\geqslant

\title{Short and local transformations between ($\Delta+1$)-colorings\thanks{This work was supported by ANR project GrR (ANR-18-CE40-0032)}}
\author[1]{Nicolas Bousquet}
\author[1]{Laurent Feuilloley}
\author[2]{Marc Heinrich}
\author[3]{Mikaël Rabie}

\affil[1]{Univ. Lyon, CNRS, INSA Lyon, UCBL, LIRIS, UMR5205, F-69621, Lyon, France. nicolas.bousquet@univ-lyon1.fr and laurent.feuilloley@univ-lyon1.fr}
\affil[2]{School of Computing, University of Leeds, UK. m.heinrich@leeds.ac.uk}
\affil[3]{Université Paris Cité, CNRS, IRIF, F-75006, Paris, France. mikael.rabie@irif.fr}

\date{}

\begin{document}

\maketitle

\begin{abstract}
    Recoloring a graph is about finding a sequence of proper colorings of this graph from an initial coloring $\sigma$ to a target coloring $\eta$. Adding the constraint that each pair of consecutive colorings must differ on exactly one vertex, one asks: Is there a sequence of colorings from $\sigma$ to $\eta$? If yes, how short can it~be?
    
    In this paper, we focus on $(\Delta+1)$-colorings of graphs of maximum degree $\Delta$. Feghali, Johnson and Paulusma proved that, if both colorings are non-frozen (i.e. we can change the color of a least one vertex), then  a quadratic recoloring sequence always exists. We improve their result by proving that there actually exists a linear transformation (assuming that $\Delta$ is a constant). 
    
    In addition, we prove that the core of our algorithm can be performed locally.
    Informally, this means that after some preprocessing, the color changes that a given node has to perform only depend on the colors of the vertices in a constant size neighborhood.
    We make this precise by designing of an efficient recoloring algorithm in the LOCAL model of distributed computing. 
    
    \medskip
    
\end{abstract}

\section{Introduction}

\subsection{Graph Recoloring and configuration graph}
\label{subsec:intro-recoloring}

A (proper) coloring of a graph is an assignment of colors to the vertices such that no two neighbors have the same color. 
Given two colorings of a given graph (referred to as the \emph{source} and \emph{target colorings}), the process of recoloring consists in starting from the source coloring and changing the color of one vertex at a time, in order to reach the target coloring, with the guarantee that the coloring is proper at all steps.


For a given graph $G$, and an integer~$k$, the two classic recoloring questions are:

\begin{question}
\label{question:possible}
Is it possible to find a recoloring between any pair of $k$-colorings of $G$? 
\end{question} 

\begin{question}
\label{question:nb-steps}
When it is possible, how many steps are needed?
\end{question}

We can restate these questions from an alternative point of view using the \emph{configuration graph}. 
For a given graph~$G$ and an integer $k$, the \emph{$k$-configuration graph} has a vertex for each proper coloring of~$G$, and an edge between every pair of colorings that differ on exactly one vertex of~$G$. 
Finding a recoloring sequence between two colorings is then the same as finding a path in the configuration graph.
From this point of view, Question~\ref{question:possible} becomes: Is the configuration graph connected? and Question~\ref{question:nb-steps} becomes: What is the diameter of the configuration graph?


\paragraph{Typical behavior when the number of colors varies.}
For any given graph, the answer to the two questions above varies with the total  number $k$ of colors (the \emph{palette size}). The typical behavior one can expect can be described by a series of regimes:

\begin{enumerate}
    \setlength{\itemsep}{0.8pt}
    \item With too few colors, proper colorings simply do not exist, hence we cannot discuss recoloring.
    \item With few colors, proper colorings exist, but for most pairs of colorings, recoloring is impossible. (That is, the configuration graph has many small connected components.) 
    \item\label{item:regime-a} With a larger palette, we reach a regime where recoloring is feasible in general, but it can take many steps. (That is, the configuration graph is connected, or not far from being connected, but it has large diameter.)
    \item By increasing the number of colors, we make the recoloring sequences shorter and shorter.
    \item\label{item:regime-b} On the extreme, with a large enough number of colors, a recoloring process exists where every vertex has its color changed only a constant number of times.
\end{enumerate}

\subsection{Palette size as a function of the maximum degree}

Many works are devoted to understand precisely when the changes of regimes occur for specific classes of graphs, or for different values of some graph parameters. 
The best-studied parameters are the degeneracy of the graph and the maximum degree.
In this paper, we focus on graphs of maximum degree $\Delta$, where the palette size is a function of~$\Delta$.
This setting has attracted a lot of interest, not only in the graph theory community, but also in the random sampling and statistical physics communities. (We will mention some results from this perspective below, but more can be found in the Related Work section.) 

Previous work have established that the most important change of regime happens between the palette sizes $\Delta+2$ and $\Delta+1$ . 
For every $k \ge \Delta+2$, the $k$-configuration graph is connected, that is, any coloring can be reached from any other by vertex recolorings. Moreover, the diameter of the configuration graph is at most $2n$~\cite{CambieCvBC22+}. Note that the diameter of the configuration graph is always at least linear in $n$, since if we consider two colorings where colors are permuted, all the vertices have to be recolored at least once.
An important conjecture in the random sampling community is that the mixing time of the Markov chain of the $(\Delta+2)$-colorings of any graph is $O(n \log n)$. In other words, given any $(\Delta+2)$-coloring of a graph, if we perform a (lazy) random walk on the set of proper $(\Delta+2)$-colorings, we should sample (almost) at random a coloring after $O(n \log n)$ steps.\footnote{This question is still widely open, and the best known upper bound on the number of colors to obtain a polynomial mixing time is $(\frac{11}{6}-\epsilon) \Delta$~\cite{ChenDMPP19}, slightly improving a classical result of Vigoda~\cite{Vigoda99}.}

In contrast, for $(\Delta+1)$ colors, the configuration graph is disconnected, in general. 
At first sight, this is a huge step for just one color less: we go from a case where we could navigate between colorings in the fastest way, even at random somehow, to a case where we cannot even reach some colorings from some others. The landscape is actually more subtle.
Indeed, Feghali, Johnson, and Paulusma~\cite{FeghaliJP16} proved that the configuration graph of the $(\Delta+1)$-colorings of a graph consists of a set of isolated vertices plus a unique component containing all the other colorings. 
Note that isolated vertices in the configuration graph correspond to colorings that are \emph{frozen}, in the sense that no change of colors can be performed. 
Bonamy, Bousquet, and Perarnau~\cite{bonamy2021frozen} proved that, if $G$ is connected, then the proportion of frozen $(\Delta+1)$-colorings of $G$ is exponentially smaller than the total number of colorings. Therefore, the configuration graph consists of a giant component of small diameter plus isolated vertices, and a long enough random walk can visit almost all $(\Delta+1)$-colorings. 


Now, the key question is: What is the diameter of the giant component? If it is small as a function of $n$, then we are in a situation close to the one of $\Delta+2$ colors, with just some special cases which are frozen colorings. If it is large, then there is a real change of regime: not only are there some isolated vertices, but it is hard to navigate between the colorings of the giant component.

In their influential paper, Feghali, Johnson, and Paulusma~\cite{FeghaliJP16} proved an upper bound of~$O(n^2)$ on the diameter of the unique non-trivial component. That is, the configuration graph does not have a large diameter, but we do not know whether it is as small as in the case of $\Delta+2$ colors.
Actually, we know that $\Theta(n^2)$ is the best possible for paths, thanks to a specific lower bound construction for $3$-colorings of paths~\cite{BonamyJ12}.
But what about graphs that are not paths? 
For general graphs, the only known lower bound is the trivial $\Omega(n)$, leaving open whether there is something special about $\Delta=2$ or whether the lower bound could be generalized. Our first theorem establishes that we are in the first situation: as soon as $\Delta$ is at least~3, the diameter of the non-trivial component drops to $\Theta(n)$ (for constant~$\Delta$).

\begin{restatable}{theorem}{ThmGlobal}
\label{thm:main-global}
Let $G$ be a connected graph with $\Delta \ge 3$ and $\sigma,\eta$ be two non-frozen $k$-colorings of $G$ with $k \ge \Delta+1$. Then we can transform $\sigma$ into $\eta$ with a sequence of at most $O(\Delta^{c \Delta}n)$ single vertex recolorings, where $c$ is a constant.
\end{restatable}

In other words, we lower the upper bound on the diameter of the non-trivial component from $O(n^2)$ to $f(\Delta)\cdot n$ and give an alternative and completely different proof of the result of~\cite{FeghaliJP16}. 
This brings $(\Delta+1)$-colorings in the category of colorings for which there exists linear transformations, a topic that has received considerable attention in recent years (see Section~\ref{sec:diam_lin}). 

An interesting direction for future work is to determine whether we can reduce the dependency in terms of $\Delta$. We actually have no lower bound that ensures that a dependency in $\Delta$ is necessary. In other words:

\begin{question}
Given $\alpha,\beta$ two non-frozen $(\Delta+1)$-colorings, is it possible to transform $\alpha$ into $\beta$ in $O(n)$ steps independent of $\Delta$?
\end{question}


Finally, note that at some steps of the proof, we can reduce the exponential dependency on $\Delta$ into a polynomial one by adapting a result of Bousquet and Heinrich~\cite{BousquetH19}, but we did not succeed to do it at every step. We thus decided to keep the proof as simple as possible.


\subsection{Configuration graphs of linear diameter}\label{sec:diam_lin}

An active line of research consists in determining which number of colors ensures that the diameter of the configuration graph is linear, in various settings beyond the bounded degree case. 
In addition to the optimality, the focus on this regime is motivated by the fact that having a linear diameter is a necessary condition to get an almost linear mixing time for the underlying Markov chain. 

Bousquet and Perarnau~\cite{BousquetP16} proved that the diameter of the $k$-configuration graph is $O(dn)$ as long as $k \ge 2d+2$ where $d$ denotes the degeneracy of the graph. Even if this proof holds by induction, it is usually hard to get linear diameters via induction proofs or vertex identification techniques, the most classic techniques in reconfiguration proofs. Several other techniques have then been used in other papers such as: discharging proofs~\cite{BartierBFHMP23,BousquetH19}, Thomassen-like approaches~\cite{DvorakF21}, or buffer sliding~\cite{BousquetB19}. 

Theorem~\ref{thm:main-global} is a contribution in this line of research since we prove that the configuration graph consists of isolated vertices plus, eventually, a connected component of linear diameter. This is, as far as we know, the first result which provides a linear diameter on the connected components of the configuration graph while the configuration graph itself is not necessarily connected. 

Our proof of Theorem~\ref{thm:main-global} introduces a new proof technique to ensure that the reconfiguration graph admits a linear diameter. We think that this proof technique is of interest and can be generalized further to other problems. The technique is related to the notions of parallelization and locality for reconfiguration, that we introduce in the next subsections. 
These have been studied recently by the distributed computing community, but as far as we know, had not been used in the more classic (sequential) reconfiguration world.

\subsection{Parallel recoloring}

\paragraph{Recoloring in parallel and dependencies between recoloring steps.}
Consider a recoloring instance where the source and the target colorings differ only on an independent set.
In this instance, any sequence created by iteratively assigning its target color to a vertex that does not have it already, is valid. 
The order can be chosen arbitrarily, because there is no dependency between the color changes: a vertex does not need one of its neighbors to first change its color in order to be able to change its own.

A way to capture this absence of dependency is to note that we can parallelize the recoloring: we can just take all the vertices that do not have their target colors and recolor them in parallel. Note that we need to be careful with the notion of parallel recoloring: after all, for any recoloring task, we could just say ``recolor all vertices in parallel'', but this would not make sense, since we simultaneously recolor adjacent vertices. 
To make it meaningful, the standard definition consists in allowing parallel recoloring only for vertices that are not adjacent.\footnote{We will review the literature on distributed recoloring later in the paper.} In other words, at any given step  of a parallel recoloring schedule, the vertices that change color form an independent set. From such a parallel schedule, it is easy to derive a sequential recoloring schedule: decompose any parallel step by performing all the individual vertex recolorings one after the other. 

Now, one might wonder if, in general, allowing parallel recoloring dramatically reduces the number of steps or not.
Let us consider two examples with very different behaviors.

\paragraph{Paths with $\Delta+1$ colors.}
Consider the case of paths with $3$ colors (note that $\Delta=2$ for paths). In particular, consider a source coloring  of the form 1,2,3,1,2,3... and a target coloring of the form 2,3,1,2,3,1... 
In this case, even if we allow parallelization, at step~$i$, only the vertices at distance at most~$i$ from an extremity can change color. 
Therefore, the recoloring must be very sequential, and we will use at least $\Omega(n)$ parallel steps. In other words, there are strong dependencies between color changes\footnote{Actually in that case, one can prove that a recoloring sequence needs $\Theta(n^2)$ single vertex recolorings and that we can recolor it with  $\Theta(n)$ parallel steps.}.

\paragraph{General graphs with $2\Delta+2$ colors.}
Now for $2\Delta+2$ colors, the behavior is completely different. We will illustrate this by designing an algorithm producing a very short parallel recoloring schedule.
The algorithm is based on the idea of partitioning the palette of colors into two smaller palettes of size $\Delta+1$ each, \emph{palette A} ($a_1,...,a_{\Delta+1}$) and \emph{B} ($b_1,...,b_{\Delta+1}$). The key observation is: in any proper coloring, for any vertex with a color from palette $A$ (resp. $B$), we can find a new non-conflicting color in palette $B$ (resp. $A$), because the smaller palettes are large enough. See Algorithm~\ref{algo:parallel}.
 
\begin{algorithm}
    \begin{algorithmic}
\For{$i$ in [$b_1,...,b_{\Delta+1}$]} 
    \State At Step $i$: every vertex with source color $i$ takes a new non-conflicting color in palette $A$.
\EndFor
\For{$i$ in [$b_1,...,b_{\Delta+1}$]} 
    \State At Step $i+\Delta+1$: every vertex with target color $i$ takes its target color.
\EndFor
\For{$i$ in [$a_1,...,a_{\Delta+1}$]} 
    \State At Step $i+2\Delta+2$: every vertex with target color $i$ takes a new non-conflicting color in palette $B$.
\EndFor
\For{$i$ in [$a_1,...,a_{\Delta+1}$]} 
    \State At Step $i+3\Delta+3$: every vertex with target color $i$ takes its target color.
\EndFor
\end{algorithmic}
\caption{Generating a parallel schedule for $2\Delta+2$ colors.}\label{algo:parallel}
\end{algorithm}


Note that the color changes happening at the same step are performed by sets of vertices that are independent, since they are color classes either in the source or the target coloring. 

This algorithm produces a schedule that uses $4(\Delta+1)$ recoloring steps, which is very small when compared with the (at least) $n$ steps that are necessary for sequential recoloring in general, and for the parallel recoloring of paths in the previous paragraph. 

This naturally leads to the following third question.

\begin{question}
\label{question:nb-paraellel-steps}
When recoloring is possible, how many \emph{parallel steps} are needed?
\end{question}

\paragraph{Parallel recoloring with $(\Delta+1)$ colors.}
The proof of Feghali, Johnson, and Paulusma~\cite{FeghaliJP16} is based on an identification technique and needs $O(n)$ parallel steps. And one can naturally see that this result cannot be improved since, if we take the power of a path (which can be chosen of arbitrarily large degree), we can have an almost frozen coloring except on the boundaries (as in the case of paths with $3$ colors). So, in order to recolor a vertex in the middle, we first have to recolor a long chain of vertices starting in the border of the graph. However, what we succeeded to prove is that this large number of parallel steps are actually needed in order to unfreeze vertices almost everywhere. And we prove that we can do it efficiently. In other words, the following holds:

\begin{theorem}[Updated version of Theorem~\ref{thm:main-global}]
For any connected graph of maximum degree $\Delta \geq 3$ and $k \ge \Delta+1$, we can transform any non-frozen $k$-coloring into any other with a sequence of:
\begin{itemize}
    \item at most $O(n)$ single vertex recolorings followed by,
    \item a parallel schedule of length at most $O(f(\Delta))$.
\end{itemize}
\end{theorem}

Actually our result is even better since we can ensure that, if, in the initial and target colorings, for each vertex there is a non-frozen vertex close enough, we can simply remove the first part and only recolor with a parallel schedule of length at most $O(f(\Delta))$.


\subsection{Distributed algorithms and locality}

\paragraph{From parallel recoloring to locality}

When studying Algorithm~\ref{algo:parallel}, it appears that not only it produces a short parallel schedule, but it also has a strong \emph{locality} property.
More precisely, let the recoloring schedule of a given vertex be the series of color changes it has to take, along with the appropriate time step.
We claim that the recoloring schedule of any fixed vertex is independent of the vertices that are outside a ball of size $4(\Delta+1)$ around it.
That is, changing the source/target colors of far away nodes, or even the topology of the graph far from the vertex would not change anything from the point of view of the vertex.

To see this, consider a node $v$ with color $b_2$. At stage 1 of the algorithm, it has to choose a non-conflicting color in palette $A$. As a consequence, its color depends on the colors of its neighbors at that moment. These can be of two types: either they are the source colors of the vertices, or they are colors of vertices with source color $b_1$, that have changed at the previous step, based on their neighbors (which all had their source colors). Therefore, the new color only depends on the initial colors of the vertices at distance at most two. 
Our claim follows from the iteration of this observation.

\paragraph{Locality and LOCAL model}

Notions of locality with the flavor described above have been studied for a long time in the theory of distributed computing, in what is called the \emph{LOCAL model}. 
There, a typical question is: suppose we let every vertex know its neighborhood at distance $\ell$, can it choose an output such that the collection of individual outputs makes sense globally? In our setting, this translates to: if every vertex knows its neighborhood at distance $\ell$, can it produce a schedule such that globally the schedule is a proper parallel recoloring? This distance $\ell$ is called the \emph{number of rounds} in the LOCAL model, and can be called the locality of the task. Hence, the following question:

\begin{question}
\label{question:locality}
What is the locality of recoloring?
\end{question}

Our example of $2\Delta+2$ colors was useful to introduce the notion of parallel schedule length and of locality, but it can be misleading because the two are basically equal in this case. 
This is because, at each step a node can check the colors of its neighbors and update its own, performing both the computation of the schedule and the application of it. 
In general, there is no such equality. It can be that the locality is larger, because the vertices need to look far to be able to produce a proper schedule (in particular, with no two adjacent nodes changing color at the same time). It can also be that the schedule is larger, for example it can be larger than $n$ (when the recoloring sequence is super polynomial for instance) whereas the locality can never be larger than $n$, since within this distance the vertex can see the whole graph.
(There are some subtleties here related to symmetries, that will be clarified when defining properly the LOCAL model.)

\paragraph{Back to $\Delta+1$ colors}

In general, when we consider a non-frozen $(\Delta+1)$-coloring of a graph $G$, there cannot exist a short parallel schedule, and the recoloring cannot be very local. This is because there might exist a unique non-frozen vertex, and the ``non-frozen-ness'' might have to be propagated edge by edge to the rest of the graph, in a similar way as for 3-colored paths.
What we prove is that it is the only case where the transformations between $(\Delta+1)$-colorings have to be global.
That is, if non-frozen vertices are well-spread, then we can actually compute a recoloring sequence locally.
Our second main theorem is the following:\footnote{For a formal definition of LOCAL model and of r-locally non-frozen colorings, the reader is referred to Section~\ref{sec:prelim}.}

\begin{restatable}{theorem}{ThmLocal}
\label{thm:main-local}
Let $G$ be a graph with $\Delta \ge 3$, $k \in \mathbb{N}$ with $k \ge \Delta+1$. Let $\sigma,\eta$ be two $k$-colorings of $G$ which are $r$-locally non-frozen. There exists three constants $c,c',c''$ such that we can transform $\sigma$ into $\eta$ with a parallel schedule of length at most $O(k^{c\Delta}+\Delta^{c'}r)$. Moreover, this schedule can be found in
\begin{itemize}
    \item $O(\Delta^{c''}+\log^* n+k)$ rounds if $k \ge \Delta+2$.
    \item $O(\Delta^{c''}+\log^2n\cdot \log^2\Delta+k)$ rounds otherwise. 
\end{itemize}
\end{restatable}

Informally, the number of rounds we need in the LOCAL model to provide a distributed recoloring sequence can be seen as how much we need to understand the graph globally to provide a recoloring sequence. When we look into our proof, the $\log^* n$ (or $\log^2n$) term in the number of rounds is there to compute a maximal independent set at distance $\Omega(1)$ (or a $\Delta+1$ coloring). If we are given such colorings and independent sets, the number of rounds is independent of $n$.

\paragraph{Impact and open questions for distributed recoloring.}
Our results are also interesting from the distributed computing point of view, since they improve on the state of the art of distributed recoloring in several ways. 
Theorem~\ref{thm:main-local} directly improves some results of~\cite{BonamyORSU18} on distributed recoloring. 
One problem studied in~\cite{BonamyORSU18} consists in recoloring $3$-colored graphs of maximum degree $3$ with the help of an extra color. They provide an algorithm that finds a parallel schedule of length $O(\log n)$ in a polylogarithmic number of rounds in the LOCAL model. Theorem~\ref{thm:main-local} implies that a constant length schedule can be found in $O(\log^*n)$ rounds (and it holds even if we start from an arbitrary locally non-frozen $4$-colorings instead of $3$-colorings plus an additional color). Theorem~\ref{thm:main-local} also directly  solves two open questions from~\cite{BonamyORSU18}: 
\begin{itemize}
    \item The first question is about the complexity of finding a schedule to recolor a $\Delta$-coloring with an extra color. Theorem~\ref{thm:main-local} gives an algorithm that finds a parallel schedule of length $f(\Delta)$ in $O(F(\Delta)\log^*n)$ communication rounds (since these colorings can be seen as non-frozen $(\Delta+1)$-colorings).
    \item The second question deals with the case of $4$-colored toroidal grids with an extra color. We provide an algorithm with a constant length schedule after $O(\log^*n)$ rounds. 
\end{itemize} 

We leave as an open problem whether a schedule can be found even faster. In particular, we conjecture that, in the case of toroidal grids, such a schedule could be found in $O(1)$ communication rounds, by using the input and target colorings as symmetry-breaking tools. More generally, we were not able to answer that question:

\begin{question}
Does there exist a non-constant function $f$ such that an algorithm computing a recoloring schedule in the LOCAL model between any pair of $28$-locally non-frozen $(\Delta+1)$-colorings takes $\Omega(f(n))$ communication rounds.
\end{question}

Note that a lower bound result of this flavor can be found in~\cite{Censor-HillelR20} for the problem of maximal independent set reconfiguration but we did not manage to adapt it in our setting.

\subsection{Related work}

In this section, we focus on recoloring literature.
For references about the larger field of reconfiguration, the reader is referred to the two recent surveys on the topic~\cite{Nishimura17,Heuvel13}. 

\paragraph{Markov chain motivation}
A major motivation to study the configuration graph of colorings is the importance of this object for random sampling. The diameter of the configuration graph is a straightforward lower bound on the mixing time of the underlying Markov chain, which corresponds to sampling colorings by local changes. 
Since proper colorings correspond to states of the anti-ferromagnetic Potts model at zero temperature, Markov chains related to graph colorings received considerable attention in statistical physics and many questions related to the ergodicity or the mixing time of these chains remain widely open (see e.g.~\cite{ChenDMPP19,frieze2007survey}).

\paragraph{Recoloring graphs with other bounded parameters}
So far we have considered graphs where the degree is bounded, since it is the setting of our results, but let us quickly mention results in classes where other parameters are bounded.  
Bonsma and Cereceda~\cite{BonsmaC09} proved that there exists a family $\mathcal{G}$ of graphs and an integer $k$ such that, for every graph $G \in \mathcal{G}$, there exist two $k$-colorings whose distance in the $k$-configuration graph is finite and super-polynomial in $n$.
Cereceda conjectured that the situation is different for degenerate graphs.
A graph $G$ is \emph{$d$-degenerate} if any subgraph of $G$ admits a vertex of degree at most~$d$. In other words, there exists an ordering $v_1,\ldots,v_n$ of the vertices such that for every $i \le n$, the vertex $v_i$ has at most $d$ neighbors in $v_{i+1},\ldots,v_n$.
It was shown independently in~\cite{dyer2006randomly} and~\cite{Cereceda09} that for any $d$-degenerate graph $G$ and every $k \geq d+2$, $\cG(G,k)$ is connected. 
However, the (upper) bound on the $k$-recoloring diameter given by these constructive proofs is of order $c^{n}$ (where $c$ is a constant). 
Cereceda~\cite{Cereceda} conjectured that the diameter of $\cG(G,k)$ is of order $\mathcal{O}(n^2)$, as long as $k \ge d+2$. If correct, the quadratic function is sharp, even for paths or chordal graphs as proved in~\cite{BonamyJ12}. The best known bound on this conjecture is due to Bousquet and Heinrich~\cite{BousquetH19}, who proved that the diameter of $\cG(G,k)$ is $n^{d+1}$. The conjecture is known to be true for a few graph classes, such as chordal graphs~\cite{BonamyJ12} and bounded treewidth graphs~\cite{BonamyB18,Feghali19}.

\paragraph{Distributed recoloring}
Distributed recoloring in the LOCAL model was introduced in~\cite{BonamyORSU18}, and implicitly studied before  in~\cite{panconesi1995local}. In~\cite{BonamyORSU18}, the authors focus on recoloring $3$-colored trees, subcubic graphs and toroidal grids, and in~\cite{panconesi1995local}, the focus is on transforming a $(\Delta+1)$-coloring into a $\Delta$-coloring.
More recently, \cite{BousquetFHR21} designed efficient distributed recoloring for chordal and interval graphs. 
A few reconfiguration problems different from coloring have been studied in the distributed setting, including vertex cover~\cite{Censor-HillelMP22}, maximal independent sets~\cite{Censor-HillelR20}, and spanning trees~\cite{GuptaKP22}.

\subsection{Organization of the paper}
This introduction described the motivation and big picture.
In Section~\ref{sec:prelim}, we give the definitions needed in the rest of the paper. In Section~\ref{subsec:proof-outline}, we sketch the proof techniques. Sections~\ref{sec:locallysafe} and \ref{sec:recoloring-non-froz} are devoted to the full proofs of the results.

\section{Preliminaries}\label{sec:prelim}

\paragraph{Classic graph definitions.}
All along the paper $G=(V,E)$ denotes a graph, $n$ is the number of vertices (i.e. $n=|V|$), and $k$ is a positive integer. For standard definitions and notations on graphs, we refer the reader to~\cite{Diestel}. 
Let $G$ be a graph and $v$ be a vertex of $G$. We denote by $N(v)$ the set of \emph{neighbors} of $v$, that is the set of vertices adjacent to $v$. 
The \emph{degree} of a vertex is the number of neighbors it has, and $\Delta(G)$ the maximum degree of~$G$ (often denoted $\Delta$ for short).
The set $N[v]$, called the \emph{closed neighborhood of~$v$}, denotes the set $N(v) \cup \{ v \}$. Given a set $X$, we denote by $N(X)$, the set $(\cup_{v \in X} N(v)) \setminus X$. The \emph{distance} between $u$ and $v$ in $G$ is the length of a shortest path from $u$ to $v$ in $G$ (by convention, it is $+ \infty$ if no such path exists), and it is denoted by $d(u,v)$. Let $r \in \mathbb{N}$. 
We denote by $B(v,r)$ the \emph{ball} of center $v$ and radius $r$, which is the set of vertices at distance at most $r$ from~$v$. 
A vertex $w$ belongs to the \emph{boundary} of $B(v,r)$ if the distance between $v$ and~$w$ is exactly $r$. The \emph{interior} of a ball $B$ is the ball minus its boundary (\emph{i.e.} $B(v,r-1)$ for a ball $B(v,r)$, with $r>0$).



\paragraph{Frozen and safe recoloring.}
Let $c$ be a coloring of $G$. A vertex $v$ is \emph{frozen} in $c$ if all the colors appear in~$N[v]$.
The coloring $c$ is \emph{frozen} if all the vertices are frozen. 
Note that a frozen coloring is an isolated vertex of the configuration graph.

Let $\alpha$ be a coloring of $G$, and $X$ be a subset of vertices. We denote by $G[X]$ the subgraph of $G$ induced by $X$, and by $\alpha_X$ the coloring $\alpha$ restricted to the vertices of $X$.
We say that two colorings $\alpha$ and $\beta$ agree on $X$ if $\alpha_{X} = \beta_{X}$.

A \emph{recoloring step} consists in changing the color of a non-frozen vertex to one that does not appear in its neighborhood. In a \emph{recoloring by independent sets}, instead of changing the color of one vertex at each step, we are allowed to change the color of an independent set of non-frozen vertices (while keeping a proper coloring).

We introduce three new definitions: \emph{r-locally non-frozen colorings}, \emph{$r$-safe graphs} and \emph{ladders}.  

\begin{definition}
A coloring is \emph{$r$-locally non-frozen} if, for every vertex $v$, there exists a non-frozen vertex at distance at most $r$ from $v$.
\end{definition}

\begin{definition}
The $k$-colorings of a graph $G$ are \emph{$r$-safe} if, for every vertex $v$, and every $k$-coloring where $v$ is non-frozen, the following holds. For any vertex $w$ at distance $r$ from $v$, there exists a recoloring sequence, such that: $w$ is recolored, all the other recolored vertices are in the interior of $B(v,r)$, and $v$ is again non-frozen at the end of the sequence.
\end{definition}

Since we only consider $(\Delta+1)$-colorings in the paper, we will say that $G$ is \emph{$r$-safe} if the $(\Delta+1)$-colorings of $G$ are $r$-safe.

The last definition we introduce originates from the following remark. 

\begin{remark}
\label{rk:unfreeze-neighbor}
Consider a non-frozen vertex $u$ in a $\Delta+1$-coloring of the graph. If we change its color, then all its frozen neighbors become unfrozen. 
\end{remark}

Indeed, before the change, for any frozen neighbor $v$ of $u$, all the colors appear exactly once in $N[v]$ (because we consider $\Delta+1$ colors). Thus, after the change, the old color~$c$ of $u$ does not appear anymore in $N[v]$, and $v$ has two possible colors: its current color and~$c$.
Now, let us go one step further. Suppose that $v$ had another neighbor $z$, not adjacent to~$u$, that was also frozen at the beginning. The recoloring of $u$ keeps $z$ frozen, but then the  recoloring of $v$ with color $c$ unfreezes it. 
By iterating this process, we get what we call a \emph{ladder}.

\begin{definition}
Given an induced path $P$ where the first vertex in the path is non-frozen, and all the other vertices are frozen, \emph{a ladder} consists in recoloring all the vertices of $P$ one by one.
\end{definition}

Note that at the end of the sequence, the other endpoint $w$ of $P$ has changed color, and it is non-frozen.
Moreover, for every consecutive pair of vertices $v_iv_{i+1}$ in the path, where $v_i$ appears first between $v$ and $w$, the final color of $v_{i+1}$ is the initial color of $v_i$.

\paragraph{\textsc{LOCAL} model and distributed recoloring}

The LOCAL model is a classic model of distributed computing (see the books and surveys \cite{BarenboimE13, Peleg00, suomela2013survey}). The model consists of a graph $G=(V,E)$ where each vertex $v\in V$ has a unique identifier (\emph{e.g.} an integer in $[1,n^2]$). 
Each edge corresponds to a communication link between two vertices. 
Initially, each vertex only has access to its identifier and its list of neighbors. 
The computation proceeds in rounds, and each vertex communicates synchronously with its neighbors at each round. 
We say that an algorithm \emph{runs in $r$ rounds} if every vertex can run this algorithm for $r$ rounds and then gets an output.
Since we do not impose any limit on the memory, on the information propagated along a link at each round, or the computational power of the vertices, this is equivalent to a model where every vertex knows its full neighborhood at distance at most $r$, and then chooses an output. 

Usually, the validity of the output can be checked locally. For example, to check that a coloring is proper, a vertex just needs to check that its color output by the algorithm differs from its neighbors'. An algorithm is valid if its output is correct. In what follows, we will only consider valid algorithms.

Distributed recoloring in the LOCAL model is defined as follows. Each vertex $v$ is given as input its initial color $c_0$ and its target color $c_{end}$. 
It outputs a schedule $c_{0},c_{1},\ldots,c_{\ell}=c_{end}$ of length $\ell$, which is the list of colors taken by $v$ all along the transformation. The output is correct if this schedule corresponds to a recoloring by independent sets.
In one communication round, each vertex can check that the schedule is consistent by checking that at each step: (i) its color differs from its neighbors', and (ii) if its color changes at some step $i>0$ (i.e. $c_{i-1}\neq c_i$), then the color of none of its neighbors is modified at that same step. 

When we handle $r$-locally non-frozen colorings in the distributed setting, a vertex is given as input its distance to a closest unfrozen vertex in both the initial and target colorings. The input validity can be checked in one round, as each vertex just needs to check that (i) both colorings are locally proper (around its vertex), and that (ii) it is unfrozen if the integer assigned to it is $0$, or it has a neighbor at a smaller distance if its distance is positive.

\section{Outline of the proofs}
\label{subsec:proof-outline}

The proofs of both our Theorems~\ref{thm:main-global} and~\ref{thm:main-local}  are in two steps.
The first step is slightly different, but the second step is the same for both results.
%
\paragraph{First step.}

The first step consists in reaching a coloring where the vertices of a fixed set $I$ are all non-frozen. 
For Theorem~\ref{thm:main-global} (centralized recoloring), this step corresponds to the following proposition, where we start from a non-frozen coloring.

\begin{restatable}{proposition}{PropCentralizedWarming}
\label{prop:centralized-warming}
Let $G$ be a connected graph of maximum degree $\Delta \ge 3$, and let $I$ be a maximal independent set at distance $d\geq 15$ in $G$. Let $\sigma$ be a coloring of $G$ that is non-frozen. Then it is possible to transform $\sigma$ into a coloring $\mu$ where $I$ is non-frozen, with $O(n)$ single vertex recolorings.
\end{restatable}

For Theorem~\ref{thm:main-local} (distributed recoloring), the first step corresponds to the following proposition, where we start from an $r$-locally non-frozen coloring. 

\begin{restatable}{proposition}{PropDistributedWarming}
\label{prop:distributed-warming}
Let $G$ be a connected graph of maximum degree $\Delta \ge 3$, and let $I$ be a maximal independent set at distance $d\geq 15$ in $G$. 
Let $\sigma$ be an $r$-locally non-frozen coloring. 
Then it is possible to transform $\sigma$ into a coloring $\mu$ where $I$ is non-frozen, in $O(d\Delta^{4d+10}+d\log^*n+r)$ rounds, and with a schedule of length $O((r+d)d\Delta^{6d+10})$ in the LOCAL model.
\end{restatable}

Actually, the proofs of both Proposition~\ref{prop:centralized-warming} and~\ref{prop:distributed-warming} will use as an essential building block the following theorem, which is of independent interest.

\begin{restatable}{theorem}{ThmLocalDuplication}
\label{thm:local-duplication}
For every $r \geq 7$ and every graph $G$ of maximum degree $\Delta \ge 3$, $G$ is $r$-safe.
\end{restatable}

Remember that this means that given a $(\Delta+1)$-coloring, if a vertex $x$ is non-frozen, and $y$ is a vertex at distance $7$ from $x$, then we can recolor the vertices of $B(x,7)$ such that:
\begin{enumerate}[label=(\roman*)]  \setlength\itemsep{0.5pt}
    \item at the end of the transformation, $x$ is still non-frozen,
    \item only vertices in $B(x,6)$ and $y$ are recolored,
    \item $y$ is non-frozen at the end of the transformation.
\end{enumerate}
Informally speaking, the result ensures that, in $(\Delta+1)$-colorings, we can locally ``duplicate'' non-frozen vertices.

\paragraph{Second step.}
The second step, which is common to both theorems, consists in reaching a fixed coloring~$\gamma$, and corresponds to the following proposition.


\begin{restatable}{proposition}{propMain}
\label{prop:main}
Let $G$ be a graph with $\Delta \ge 3$ and $I$ be an independent set at distance~$28$.
Let $r,k,k' \in \mathbb{N}$ such that $k' < k$, $k \ge \Delta+1$. Let $\mu, \gamma$ be two colorings, using respectively at most $k$ and $k'$ colors, that are both non-frozen on $I$. 
There is a recoloring schedule from $\mu$ to $\gamma$ of length at most $(k')^{O(\Delta)}$.
Moreover, such a recoloring schedule can be computed in $O(\Delta)$ rounds in the LOCAL model. 
\end{restatable}

Note that even if $k=\Delta+1$, a $k'$-coloring with $k'=\Delta$ exists, by Brook's theorem. Indeed, since $\sigma$ is non-frozen and $\Delta \ge 3$, $G$ is neither a clique nor an odd cycle.

We now have all the tools to establish our main results. Let us restate and prove our two main theorems. 

\ThmGlobal*

\begin{proof}
By Proposition~\ref{prop:centralized-warming}, we can transform $\sigma$ (resp. $\eta$) into a coloring $\mu$ (resp. $\mu'$) which is  non-frozen on $I$ by recoloring $O(n)$ vertices in total.

Let $\gamma$ be an arbitrary $\Delta$-coloring of $G$.
In order to build the recoloring sequence from $\mu$ to $\mu'$, we will build one from $\mu$ to $\gamma$, and one from $\gamma$ to $\mu'$. 
By Proposition~\ref{prop:main}, there is a recoloring schedule from $\mu$ (resp. $\mu'$) to $\gamma$ recoloring at most $\Delta^{O(\Delta)}$ independent sets. This sequence recolors at most $\Delta^{O(\Delta)}$ times each vertex, which completes the proof.
\end{proof}

The second theorem is about local reconfiguration, and we assume that the colorings are $r$-locally non-frozen.

\ThmLocal*

\begin{proof}
We first compute an independent set at distance $d=28$ in a distributed manner in time $O(\Delta^{28}+\log^* n)$ rounds in the LOCAL model by~\cite{BarenboimEK14}.
Then by plugging the result of Proposition~\ref{prop:distributed-warming} with $d=28$, we can transform $\sigma$ (resp. $\eta$) into a coloring $\mu$ (resp. $\mu'$) such that all the vertices of $I$ are non-frozen with a recoloring schedule of length $O(r \Delta^{178})$ in $O( \Delta^{122}+\log^* n+r)$ rounds. 

Assume first that $k \ge \Delta+2$. 
It is easy to transform $\mu$ into a coloring $\gamma$ with $k-1$ colors in one round: for every vertex that has color $\Delta+2$, move it to a color of smaller index. 
Such a color must exist, and the transformation takes only one round.
Now by Proposition~\ref{prop:main}, we can transform $\mu'$ into $\gamma$ efficiently, and finish this proof.

Now, if $k = \Delta+1$, we first compute an arbitrary $\Delta$-coloring, in time $O(\log^2n\log^2\Delta)$, using the algorithm of~\cite{GhaffariK22},
and then use Proposition~\ref{prop:main} twice (between $\mu$ and $\gamma$, and between $\mu'$ and $\gamma$).
\end{proof}



\section{Safeness and consequences}
\label{sec:locallysafe}

\subsection{Maximum degree at least 3 ensures safeness}

The goal of this section is to prove the following theorem, that ensures that, in any ball with an unfrozen vertex, we can unfreeze at a vertex of its border while keeping its center unfrozen. 

\ThmLocalDuplication*

Consider a graph $G$ of maximum degree $\Delta \ge 3$. 
Let $\sigma$ be proper $(\Delta+1)$-coloring of~$G$, and $v$ be a non-frozen vertex. Let $B = B(v, r)$.
Note that if the boundary of $B$ is empty (that is, the whole graph is contained in $B(v, r-1)$) then $G$ is $r$-safe. For the rest of the section, we will assume that this is not the case.

Let $w$ be a vertex of the boundary of $B$. Our goal is to prove that there exists a recoloring sequence of the vertices of the interior of $B$ plus $w$, which recolors $w$, and such that at the end of the sequence, $v$ is still non-frozen. Moreover, we will show that this recoloring sequence recolors each vertex at most twice (and at most $2r$ vertices in total).

We will call \emph{nice} such a recoloring sequence. The existence of a nice recoloring sequence implies Theorem~\ref{thm:local-duplication}. Let us first give some conditions which ensure the existence of a nice recoloring sequence.

\begin{lemma}
\label{lem:Pnonfrozen}
Let $P$ be a shortest path from $v$ to $w$. Assume that $P$ contains a non-frozen vertex not in $N[v]$. Then there is a nice recoloring sequence.
\end{lemma}
\begin{proof}
Let $z$ be  the non-frozen vertex of $P$ closest to $w$. By assumption, we know that~$z$ is not adjacent to~$v$. Let~$P'$ be the subpath from~$z$ to~$w$. 
We can recolor $w$ by recoloring a ladder along this path~$P'$. Let us check that this is a nice recoloring sequence. 
All the vertices of $P'$, except $w$, are in the interior of $B$, because $P$ is a shortest path from the center of the ball $B$ to $w$. Moreover, after this transformation $v$ is still non-frozen since none of its neighbors were recolored. 
Finally, every vertex is recolored at most once.
\end{proof}

We can extend this property to the vertices at distance 1 from the path $P$.

\begin{lemma}
\label{lem:Pneighbors}
Let $P$ be a shortest path from $v$ to $w$. Assume that there is a non-frozen vertex~$z$ adjacent to $P$, such that $3\leq d(v,z) \leq r-1$. Then there is a nice recoloring sequence.
\end{lemma}
\begin{proof}
The argument is similar to the one of Lemma~\ref{lem:Pnonfrozen}. Let $z$ be a vertex satisfying the conditions of the lemma, that is the closest to $w$. Note that  $z$ is in the interior of $B$, since $d(v,z) \le r-1$. Let $z'$ be the neighbor of $z$ in $P$ which is the closest to $w$, then $z'$ is at distance at least $2$ from $v$, in particular, it is not a neighbor of~$v$. Then, we can again recolor along a ladder that starts with $z$, $z'$, and then continues along $P$ towards~$w$. This allows us to recolor $w$ while leaving the neighbors of~$v$ and the boundary of $B$ untouched. Each vertex is recolored at most once, which implies that this is a nice recoloring sequence.
\end{proof}

\begin{lemma}
\label{lem:Pperiodic}
Let $P = v_0, \ldots, v_r$ be a shortest path from $v$ to $w$. If there is an index $2 \le i \leq r - 3$ , such that $\sigma(v_i) \neq \sigma(v_{i+3})$, then there is a nice recoloring sequence.
\end{lemma}

\begin{proof}
By Lemma~\ref{lem:Pnonfrozen}, we can assume that all the vertices of $P$, except for $v=v_0$ and its neighbor $v_1$, are frozen. 
Let us denote by $\eta$ the coloring obtained by recoloring the ladder along $P$, starting either from $v$, if $v_1$ is frozen, or $v_1$, if it is non-frozen, and ending in $w$. 
In $\eta$, we have recolored $w$, but now $v$ might be frozen. If $v$ is not frozen, we are done. If $v_1$ is non-frozen, then again we are done, since we can make a ladder with just $v_1$ and~$v$. Thus, let us assume that both $v_1$ and~$v$ are frozen in $\eta$.

Amongst the indices $2 \le i \leq r - 3$ such that $\sigma(v_i) \neq \sigma(v_{i+3})$, let $i$ be the minimum one. We have the following claim:
\begin{claim}\label{clm:ii+3}
The vertex $v_{i+2}$ is non-frozen in the coloring $\eta$.
\end{claim}
\begin{proof}
Let $c=\sigma(v_{i+3})$. 
Let us make a few remarks: 
\begin{enumerate}\setlength{\itemsep}{0pt}
    \item $\sigma(v_{i+2})\neq c$, because $\sigma$ is a proper coloring,
    \item $\sigma(v_{i+1})\neq c$, because $v_{i+2}$ is frozen in $\sigma$. More generally, none of the neighbors of $v_{i+2}$ except $v_{i+3}$ has color~$c$.
    \item  $\sigma(v_i) \neq c$, because $\sigma(v_i) \neq \sigma(v_{i+3})$ by assumption.
\end{enumerate}

Now, by construction and by the properties of ladders, we have $\eta(v_{j+1}) = \sigma(v_j)$, for every vertex $v_j$ of the ladder, except $v_r=w$. 
Transposing the remarks above about $\sigma$ to $\eta$ we get that: 

\begin{enumerate}\setlength{\itemsep}{0pt}
    \item $\eta(v_{i+3})\neq c$,
    \item $\eta(v_{i+2})\neq c$, and more generally, no neighbor of $v_{i+2}$ has color $c$.
    \item $ \eta(v_{i+1}) \neq c $.
\end{enumerate}

Consequently, $c$ does not appear in the closed neighborhood of $v_{i+2}$ in $\eta$, which implies that $v_{i+2}$ is non-frozen in $\eta$, as claimed.
\end{proof}

By Claim~\ref{clm:ii+3}, $v_{i+2}$ is non-frozen in $\eta$. 
We can make a new ladder in $\eta$ along the path $P$ from $v_{i+2}$ to $v$. The vertex $w$ is not recolored by this ladder, and at the end $v$ is non-frozen. 
Since every vertex is recolored at most twice, we get a nice recoloring sequence.
\end{proof}

We now have all the tools to prove that a nice recoloring sequence always exists. Let us assume that we do not fall in one of the previous cases. 
Let $P = v_0, \ldots, v_r$ be a shortest path from $v$ to $w$.
By Lemma~\ref{lem:Pnonfrozen}, all the vertices in $P$ but the first two ones are frozen. 
By Lemma~\ref{lem:Pneighbors}, all the neighbors of $P$ that are at distance at least three from $v$ are frozen. 
Free to rename colors, Lemma~\ref{lem:Pperiodic} ensures that $\sigma(v_i) = i \mod 3$ for every $i \ge 2$. 
We denote by $\eta$ the coloring obtained by recoloring the ladder along~$P$ starting either from $v$, if $v_1$ is frozen, or from $v_1$ otherwise.
As before, at that point we are done except if both $v$ and $v_1$ are frozen in $\eta$.
Note that, for $i \ge 3$, $\eta(v_i) = (i -1) \mod 3$, because of the color shift of the ladder. 

Let us consider the vertex $v_5$. 
It cannot have degree $2$, because it is frozen in $\sigma$, and no degree-2 vertex can be frozen in a $\Delta+1$-coloring, with $\Delta\ge3$. Hence, we can assume that $v_5$ has a neighbor~$z$ outside $P$. And because $P$ is a shortest path, $z$ is at distance at least $4$ from $v$. Also note that, since we assume that $r\geq 7$, $d(v,z)\leq d(v,v_5)+1 \leq r-1$.
Therefore, by Lemma~\ref{lem:Pneighbors}, $z$ is frozen in~$\sigma$. 
We will use the following claim:

\begin{claim}
\label{clm:z-non-frozen}
If $z$ is non-frozen in $\eta$, then a nice recoloring exists.
\end{claim}

\begin{proof}
Indeed, from $\eta$, we can recolor along a ladder from $z$ to $v$. After this operation, no other vertex of the boundary is recolored, $v$~is non-frozen, and each vertex has been recolored at most twice. Hence, this defines a nice recoloring sequence.
\end{proof}

We make a case analysis depending on the number of neighbors of $z$ in~$P$.

\noindent
\textbf{Case 1:} $z$ has a single neighbor in $P$. Since $z$ is frozen in $\sigma$, $v_5$ is its only neighbor colored with $\sigma(v_5)$. In $\eta$, $v_5$ is recolored with a different color, which implies that $z$ is no longer frozen in $\eta$. By Claim~\ref{clm:z-non-frozen}, the conclusion follows.
\smallskip

\noindent
\textbf{Case 2:} $z$ has exactly two neighbors in $P$. Let $c_1$ and $c_2$ be the colors of these two neighbors in~$\sigma$. Since $z$ is frozen in~$\sigma$,  it does not have two neighbors colored with the same color. 
Moreover, in $\eta$, the two neighbors of $z$ in $P$ have color $c'_1 = c_1 -1 \mod 3$ and $c'_2 = c_2 - 1 \mod 3$ by Lemma~\ref{lem:Pperiodic} (since $z$ is incident to $v_5$, the other neighbor is at least $v_3$). Then we have $\{ c_1, c_2\} \neq \{ c'_1, c'_2\}$. It follows that $z$ is non-frozen in $\eta$, and the result follows from Claim~\ref{clm:z-non-frozen}.
\smallskip

\noindent
\textbf{Case 3:} $z$ has at least three neighbors in $P$. Since~$P$ is a shortest path, $z$ has exactly three neighbors in~$P$, and these neighbors are consecutive in~$P$. Let $3\leq i \leq 5$ such that~$v_i, v_{i+1}, v_{i+2}$ are the neighbors of~$z$ in~$P$. Since $z$ is adjacent to $v_{i+1}$, we have $\sigma(z) \neq \sigma(v_{i+1}) = (i+1) \mod 3$. Let $P'$ be the path obtained from $P$ by replacing $v_{i+1}$ by $z$. (Note that $z$ is not in the boundary of $B$, and then $z \ne w$.) 
Then $P'$ is a shortest path from $v$ to $w$, and since $\sigma(z) \neq (i+1) \mod 3$, we can apply Lemma~\ref{lem:Pperiodic} on~$P'$ to conclude. 
More precisely, if $i=3$, then $(i+1)+3$ is at distance at most $r$ because $r\geq 7$, and if $i=4$ or $5$, then $(i+1)-3\geq 2$, thus in both cases the lemma applies.

This concludes the proof, and proves Theorem~\ref{thm:local-duplication}.

\subsection{Consequences of Theorem~\ref{thm:local-duplication}}
The next lemma ensures that, in the centralized setting, we can obtain a 28-locally non-frozen coloring.

\PropCentralizedWarming*

\begin{proof}
We start by unfreezing a vertex of $I$.
Consider a pair of vertices $u,v$ such that $u$ is non-frozen and $v$ is in $I$, that minimize the distance $d(u,v)$. Note that if $u=v$,  we are done. Otherwise, we take a shortest path from $u$ to $v$ and build a ladder along this path to unfreeze $v$.

We construct an auxiliary graph $H$, where $V(H)=I$, and we put an edge $(i,i')$ in $H$ if there exists a path of length at most $2d$ from a vertex of $B(i,7)$ to a vertex of $B(i',7)$ in $G$ which does not contain any vertex in $B(i'',7)$ for any $i'' \ne i,i'$. Note that for any pair $a,b \in I$, $B(a,7)$ and $B(b,7)$ are disjoint, since $d\geq 15$.

\begin{claim}
The graph $H$ is connected.
\end{claim}

\begin{proof}
Suppose the claim does not hold. Let $A$ be a connected component of $H$. Let $i \in A$ and $j\in I\setminus A$, such that $d_G(i,j)$ is minimum among the such pairs.
If $d(i,j) \geq 2d+2$, then the vertex in the middle of a shortest path between $i$ and $j$ is at distance at least $d+1$ from any vertex in $I$ which contradicts the maximality of $I$. 
So there exists  $i \in A$ and $j\in I\setminus A$ such that $d(j,i) \le 2d+1$. Now let $x$ be the last vertex in $B(i,7)$ and $y$ the last vertex in $B(j,7)$. We have $d(x,y) \le 2d$ which gives a contradiction.
\end{proof}

Now, let us denote by $T$ a spanning tree of $H$ rooted in $v$. Let $\tau$ be a BFS ordering of $T$. The index of a vertex of $H$ is its position of appearance in the BFS. Let $i$ be the first vertex of $\tau$ that is frozen. Note that if all the vertices of $I$ are non-frozen, we are done. Also note that $i$ cannot be the root of the tree since $v$ is non-frozen. 

\begin{claim}
By recoloring a constant number of vertices, we can unfreeze $i$, and this operation leaves the vertices of index $j$ smaller than $i$ in $\tau$ non-frozen.
\end{claim}

\begin{proof}
Let $i'$ be the parent of $i$ in $T$, and let $P$ be the path  from $i$ to $i'$ in $G$ corresponding to the edge $(i,i')$ in $H$. By definition, $P$ has length at most $2d$ and does not intersect $B(i'',7)$ for any $i'' \ne i,i'$. 
Also, we can assume that $P$ is an induced path, since otherwise we can take a path on a subset of vertices of $P$, satisfying the same properties.  
If there is a vertex $y$ in $P \setminus B(i,7)$ that is unfrozen, we simply recolor a ladder from $y$ to $i$, to unfreeze $i$. Otherwise, let $x$ be the last vertex of $P$ in $B(i',7)$. By Theorem~\ref{thm:local-duplication}, by recoloring at most $14$ vertices, we can recolor $x$, leave $i$ unfrozen and while recoloring only vertices in $B(i',6)$ (and $x$). 
We can recolor a ladder from $x$ to $i$ to get the conclusion. 
In both cases, the recoloring sequence has length at most $2d+14$, and the non-frozen vertices of $I$ are kept unfrozen. \end{proof}

We iterate this construction to get all of $I$ non-frozen.
This requires at most $(2d+14) |I| \le (2d+14) n$ recoloring steps, therefore we get the result.  Note that since every vertex contains at most $\Delta^{2d}$ other vertices of $I$ at distance at most $2d$, every vertex is recolored at most $O(\Delta^{2d})$ times during the whole process.
\end{proof}

We now prove the following proposition, which is the local analogue of the previous proposition. Intuitively, it says that if we have a well-spread set of non-frozen vertices, we move it to another well-spread set locally.
 
\PropDistributedWarming*

Note that $r$ could be large and depend on $n$, in which case Proposition~\ref{prop:distributed-warming} not only moves the set of well-spread non-frozen vertices around, but also makes it more dense. 

\begin{proof}
Let $N$ be the set of non-frozen vertices at the beginning of the algorithm. We proceed in two steps: first, we show that we can somehow move the set of non-frozen vertices to a subset of $I$, and then we show how to unfreeze all the vertices of $I$.

For both steps, we will use an auxiliary coloring of the vertices of $I$. Note that this auxiliary coloring is just a tool and is independent of the coloring we are modifying. 
Let~$p$ be an integer.
Consider a graph $H$, whose vertex set is $I$ and whose edges are the pairs $(a,b)\in I$, such that $d_G(a,b)\leq p$. 
The graph $H$ has maximum degree $\Delta_H=O(\Delta^{p})$, thus we can compute a $(\Delta_H+1)$-coloring of $H$ in $O(\Delta_H + \log^*(|H|))$ rounds in $H$~\cite{BarenboimEK14}. 
Therefore, we can compute such an auxiliary coloring of $I$ in $G$ in $O(p\Delta^{p}+p\log^*n)$ rounds (in $G$). 

\begin{claim}
\label{claim:distributed-initial}
We can reach a coloring such that in the final coloring any vertex of $I$ is at distance at most $r+d$ from a non-frozen vertex of $I$, with a schedule of length $O(d\Delta^{2d+4})$ computed in $O(d\Delta^{2d+2}+d\log^*n)$ rounds.
\end{claim}

\begin{proof}
Consider an auxiliary coloring as described above, with $p=2d+2$.
Let $M_i$ be the set of vertices that are in $I$ and have received color $i$ in $H$. 
We will go through the sets $M_i$ one after another.
At step~$i$, for every $u \in M_i$ that is frozen, if $B(u,d)$ contains a vertex $v$ of $N$ that is still non-frozen, we recolor a ladder from $v$ to $u$ (where we take $v$ to be the closest non-frozen vertex).
Since, $p=2d+2$, the balls $B(u,d+1)$ with $u\in M_i$ are all disjoint by construction of the $M_i$. Therefore, we can perform this step in parallel without coordination.
Now, we want the additional property that a vertex $u$ of $I$ that has been unfrozen cannot be refrozen. This could happen if there is a non-frozen vertex in the neighborhood of $u$ that is the start of a ladder (thus at distance exactly $d$ from another vertex of $I$). 
We add a twist to the algorithm: if this situation occurs, we do not build the ladder.

To prove that the claim holds at the end of this process,
consider a vertex $w$ of $I$. 
By assumption, at the beginning $w$ was at distance at most $r$ from a non-frozen vertex $x$ of~$N$. 
Consider a vertex $u$ of $I$ in $B(x,d)$ (such a vertex exists by maximality). If this vertex $u$ is non-frozen, then the claim holds for $w$. 
If this vertex is frozen, the only possibility is that we did not build a ladder from $x$ to $u$ because of the twist in the algorithm. 
But in this case there exists a vertex $u'\in I$ in the neighborhood of $x$ which is necessarily non-frozen (since there is no obstruction to building a ladder from $x$ to $u'$).

The round complexity is dominated by the computation of the auxiliary coloring, and the schedule length can be bounded by the maximum size of a ladder inside a ball, $O(d)$ times the number of color classes $O(\Delta^{2d+4})$.
\end{proof}

\begin{claim}
\label{claim:distributed-iteration}
Consider a coloring and the distance from any vertex of $I$ to the closest non-frozen vertex of $I$ in this coloring. 
If this distance is positive, we can reach a new coloring, where this distance in strictly smaller, with a schedule of length $O(\Delta^{6d+14})$ in $O(d\Delta^{4d+10}+d\log^*n)$ rounds.
\end{claim}

\begin{proof}
Again, consider an auxiliary coloring as described at the beginning of the proof, but with parameter $p=4d+10$.
We will consider the color classes $M_i$, one after another.
For every $u \in M_i$, let $X_u$ be the ball $B(u,2d+4)$ plus the vertices of $I$ at distance exactly $2d+5$ from $u$ in $G$. 
Note that no vertex of $I$ in $V \setminus X_u$ is adjacent to $X_u$. 
If $u$ is non-frozen, then we can unfreeze all the vertices of $I \cap X_u$: since $d\geq 15$, we can proceed exactly like in the proof of Proposition~\ref{prop:centralized-warming}.
Note that, similarly to the previous proof, because of our definition of the sets $(X_u)_u$, these recolorings can be performed in parallel, and no vertex of $I$ that was non-frozen can be refrozen. 

We claim that, at the end of this recoloring, the minimum distance from any vertex $u$ of $I$ to the closest non-frozen vertex of $I$ has decreased. 
Indeed, let $v$ be the closest non-frozen vertex of $I$ from $u$ at the beginning. 
If $d(u,v) \leq 2d+4$, $u$ is non-frozen at the end of the algorithm by construction. 
Otherwise, let $x$ be the $(d+1)$-th vertex of a shortest path from $v$ to $u$. 
Note that $x$ must be at distance at most $d$ from a vertex $v'$ of~$I$. 
Thus $v'$ is in $B(v,2d+1)$. So $v'$ is unfrozen at the end of the algorithm. 
And since the distance from $u$ to $v'$ is strictly smaller than the one from $u$ to $v$, we get the condition of the claim.
The computation of the schedule length and number of rounds are similar to the ones of the previous claim, except the unfreezing of each $X_u$ uses $O(\Delta^{2d+4})$ recoloring steps.
\end{proof}

By using the algorithm of Claim~\ref{claim:distributed-initial}, and then iterating the algorithm of Claim~\ref{claim:distributed-iteration}, we can unfreeze all of $I$.

The number of iteration of Claim~\ref{claim:distributed-iteration} is at most $r+d$ by Claim~\ref{claim:distributed-initial}, thus the total schedule length is in $O((r+d)d\Delta^{6d+10})$. 
The total number of rounds is $O(d\Delta^{4d+10}+d\log^*n+r)$ since we can reuse the same auxiliary coloring for all the iterations.


\end{proof}

\section{Recoloring locally non-frozen colorings}
\label{sec:recoloring-non-froz}

The goal of this section is to prove Proposition~\ref{prop:main}. 
To do so we will first prove a few lemmas. 

\subsection{Degeneracy ordering lemma}

A graph $G$ is \emph{$d$-degenerate} if any subgraph of $G$ admits a vertex of degree at most~$d$. In other words, there exists an ordering $v_1,\ldots,v_n$ of the vertices such that for every $i \le n$, the vertex $v_i$ has at most $d$ neighbors in $v_{i+1},\ldots,v_n$.

\begin{lemma}
\label{lem:ISdecomposition}
Let $G$ be a connected $r$-locally non-frozen graph which is $k$-colorable, and let $S$ be a maximal independent set at distance at least $2r+2$. Let $B_S$ be the set of vertices at distance at most $r$ from $S$, and $G' = G \setminus B_S$. 

Then there exists a $(\Delta-1)$-degeneracy ordering of $G'$ consisting of $O(r \cdot k)$ consecutive independent sets. 
Moreover, if we are given a $k$-coloring $c$ of $G'$, such an ordering can be found in $O(r)$ rounds in the LOCAL model.
\end{lemma}
\begin{proof}
The graph $G'$ is $(\Delta-1)$-degenerate because we have removed at least one vertex from a connected graph of maximum degree $\Delta$.
 The degeneracy ordering of $G'$ will be built by first splitting $G'$ into layers such that each vertex $v$ in layer $i$ has at most $\Delta -1$ neighbors in layers $j \geq i$. Then we will split each layer into independent sets using the coloring $c$. 

We define the $i$-th layer $L_i$ of $G'$ as the set of vertices at distance exactly $i$ from $B_S$. Since $S$ is a maximal independent set at distance $2r+2$, all the vertices of $G'$ belong to a layer $i$ with $i \le r +2$. All the vertices in the first layer have a neighbor in $B_S$ and, for every $i\ge 2$, all the vertices in layer $i$ have at least one neighbor in layer $(i-1)$. So the graph induced by the layers $\cup_{j \ge i} L_j$ is $(\Delta-1)$-degenerate (and all the vertices of $L_i$ have degree $\Delta-1$ in $\cup_{j \ge i} L_j$). 
We now split each layer into $k$ independent sets using the color classes of a $k$-coloring $c$. 
We can order the vertices in the layers by color, and get a $(\Delta-1)$-degeneracy ordering of $G'$ composed of  $O(r \cdot k)$ consecutive independent sets.

Note that in the LOCAL model, if $S$ is given, computing this partition can be done in $O(r)$ rounds. Indeed, after computing its distance to $S$, each vertex knows if it is in $B_S$ or in which layer it is. As their color in $c$ is given as input, they do not need more information.
\end{proof}

\subsection{List-coloring lemma}

The following lemma is a list-coloring adaptation of a proof of Dyer et al.~\cite{dyer2006randomly} that ensures that one can transform any $(d+2)$-coloring of a $d$-degenerate graph into any other. 
Let $G$ be a graph in which, for every vertex $u$, we are given a list $L_u$ of colors. 
A coloring $c$ of $G$ is \emph{compatible} with the lists $L_u$, if the coloring is proper and for every vertex $u$, $c(u) \in L_u$. Let $\tau$ be an ordering of $V(G)$. 
We denote by $d^+_\tau(u)$ (or $d^+(u)$, when $\tau$ is clear from context) the number of neighbors of $u$ that appear after $u$ in $\tau$. We say that a set of lists is \emph{safe for $\tau$} if, for every vertex $u$, $|L_u| \ge d^+_u+2$.

We will consider particular schedules in the LOCAL model such that, at each step, all the recolored vertices are recolored from a color $a$  to a color $b$ (in particular, the recolored vertices form an independent set). We call such a reconfiguration step an \emph{$a\to b$ step}. A recoloring schedule where all the steps are $a \to b$ steps is called a \emph{restricted schedule}. Note that any schedule can be transformed into a restricted schedule by multiplying the length of the schedule by $O(k^2)$ (where $k$ is the total number of colors). Indeed, we simply have to split each step $s$ of the initial schedule into $k(k-1)/2$ different $a \to b$ steps $s_{a,b}$ for every pair of colors $a,b$. At step $s_{a,b}$, we recolor from $a$ to $b$ all the vertices recolored from $a$ to $b$ at step $s$. Note that since at step $s$, the set of recolored vertices is an independent set, all the intermediate colorings obtained after $s_{a,b}$ are proper.

\begin{lemma}
\label{lem:ISrecoloring_List}
Let $G$ be a graph, $\tau$ be an ordering of $G$ composed of $t$ consecutive independent sets and, $d = \max_{v \in V} d^+_\tau(v)$. Consider a set of lists $(L_v)_{v \in V}$ safe for $\tau$. Let $\sigma,\eta$ be two $k$-colorings of $G$ compatible with $(L_v)_{v \in V}$.

There exists a recoloring sequence from $\sigma$ to $\eta$ with a restricted schedule of length at most $k^{t+1}$ where $k =| \cup_{v \in V} L_v|$. Moreover, this schedule can be found in $O(r)$ rounds if the independent sets of $\tau$ are given.
\end{lemma}

\begin{proof}
Let $I_t,\ldots,I_1$ be the independent sets of the ordering $\tau$. For every $i \leq t$, we denote by $G_i$ the graph $G[\cup_{j \le i} I_j]$. 

Let us prove by induction on $i$ that we can recolor $G_i$ from $\sigma_{G_i}$ to $\eta_{G_i}$ with a restricted schedule of length at most $k^{i+1}$. Since $G_1$ induces an independent set, a restricted schedule of length $k \cdot (k-1)\le k^2$ exists. Indeed, for every pair $a \ne b$, we create an $a \to b$ step where we recolor the vertices of $I_1$ colored $a$ in $\sigma$ and $b$ in $\eta$ from color $a$ to color $b$. After all these steps, the coloring is $\eta_{G_1}$. Since $I_1$ is an independent set, we indeed recolor an independent set at any step. 

In order to extend the transformation of $G_{i-1}$ into a transformation of $G_{i}$ (with $i\ge2$) we perform as follows. 
For each step $s$ of the transformation of $G_{i-1}$, we will add $(k-2)$ new steps before $s$. Since the transformation is a restricted schedule, there exists $a,b$ such that $s$ is an $a \rightarrow b$ step. 
For every $c \ne a,b$, we add a $b \rightarrow c$ step, denoted $s_{b,c}$, between $s$ and the step before in the transformation of $G_{i-1}$. 
Let $I$ be the set of vertices recolored at step $s$, and $N_I$ be the set of vertices at distance exactly 1 from a vertex of $I$.
In $s_{b,c}$, we recolor all the vertices of $G_i \cap N_I$ colored $b$ with the color $c$, if it is possible (i.e. if $c$ is in their lists, and they do not have any neighbor already colored $c$) .
Note that every vertex $v$ of $I$ colored $b$ can indeed be recolored with some color $c$, distinct from $a$, since the size of the list of $v$ is at least the degree of $v$ plus two in $G_i$. 
So after these new steps, we can safely apply the $a \to b$ step without creating monochromatic edges in $G_i$.

Finally, at the end of the reconfiguration sequence of $G_{i-1}$, we add $k \cdot (k-1)$ steps in order to recolor the vertices of $I_i$ with their target colors (after $G_{i-1}$ has reached its target coloring) as we did for $I_1$. This provides a restricted schedule of length $(k-2) \cdot k^{i}+k \cdot (k-1)\le k^{i+1}$ from $\sigma$ to $\eta$ which completes the proof.

In the LOCAL model, to compute their own layers, the vertices need $O(r)$ rounds. 
In order to compute its own schedule, a vertex simulates the induction, above, which can be done with a view of $O(r)$ rounds. 
\end{proof}

As an immediate corollary, we obtain the following, where the lists are just the same $k$ colors for every vertex:
\begin{lemma}
\label{lem:ISrecoloring}
Let $G$ be a $d$-degenerate graph and $\sigma,\eta$ be two $k$-colorings of $G$ with $k \ge d+2$.
Assume that $G$ has a degeneracy ordering composed of $t$ consecutive independent sets. Then there exists a recoloring sequence from $\sigma$ to $\eta$ with a restricted schedule of size at most $k^{t+1}$ in the LOCAL model.
\end{lemma}

\subsection{Recoloring outside the balls}

Let us now prove that we can obtain a coloring where the vertices agree on $V \setminus B_S$.
Then we will explain how we can transform such a coloring into the target coloring by recoloring (almost) only vertices of $B_S$.

\begin{lemma}
\label{lem:recoloutside}
Let $k \ge \Delta+1$ and $r \geq 10$. Let $G$ be a graph of maximum degree $\Delta \ge 3$, and let $\sigma, \eta$ be two $r$-locally non-frozen $k$-colorings of $G$. Let $S$ be a maximal independent set at distance $r' \ge 2r+2$.
Let $G'=G[V \setminus B_S]$ where $B_S= \cup_{x \in S} B(x,r)$. 

Then there is a recoloring schedule of length $k^{O(r'k)}$ from $\sigma$ to $\eta'$ such that $\eta'_{G'} = \eta_{G'}$.
\end{lemma}
\begin{proof}
The first part of the recoloring sequence is a pre-processing step to ensure that every vertex $v \in  S$ is non-frozen. Since $\sigma$ is $r$-locally non-frozen, for every $v$ in $S$, there is a vertex $u$ in $B(v,r)$ such that $u$ is non-frozen. By recoloring a ladder along a shortest path from $u$ to $v$, $v$ is non-frozen. Since $B(v,r)$ does not share an edge with $B(v',r)$ for any $v,v' \in S$, we can repeat this argument for every $v \in S$ and then assume that $S$ is unfrozen. In the LOCAL model, all these recolorings pre-processing steps can be performed in parallel. So, from now on, we can assume that, in $\sigma$, every vertex of $S$ is non-frozen (and we will keep this property all along the schedule).

By Lemma~\ref{lem:ISdecomposition} and Lemma~\ref{lem:ISrecoloring}, there exists a restricted recoloring schedule $\mathcal R$ in $G'$ from $\sigma_{G'}$ to $\eta_{G'}$ in at most $(\Delta+2)^{O(r\Delta)}$ steps.

Let us now explain how we can extend the restricted schedule $\mathcal R$ of $G'$ to $G$, that is, avoid the conflicts between vertices in $G'$ and their neighbors in $G$ that are in $B_S$. 
Let $X$ be the set of vertices which are recolored during an $a \to b$ step of $\mathcal{R}$. Denote by $Y$ the set of vertices of $B_S$ such that $Y$ is adjacent to a vertex of $X$. We will recolor these vertices, before they create any conflict.
 
For each ball of radius $r$ centered in $u \in S$, we first identify the vertices of $Y_u=Y \cap B(u,r)$ that are colored $b$. Note that $Y_u$ is an independent set. 
By Theorem~\ref{thm:local-duplication}, we can recolor each vertex of $Y_u$ in at most $2r$ steps with a different color, leaving $u$ unfrozen, and without modifying the color of any other vertex in $Y_u$. Since $Y_u$ contains at  most $\Delta^r$ vertices, we can change the color of all the vertices of $Y_u$ with a schedule of length at most $2r \cdot \Delta^r $. Since all the balls of radius $r$ centered in $S$ are disjoint and do not share an edge, we can perform these schedules in parallel for each ball of radius centered in $S$.

Since the restricted schedule $\mathcal{R}$ has length at most $k^{O(r' k)}$, the new schedule have length at most $k^{O(r'k)} \cdot 2r' k^{k+2}= k^{O(r'k)}$, which completes the proof.
\end{proof}

The previous lemma ensures that, from any locally non-frozen coloring, we can obtain a locally non-frozen coloring where all the vertices but the vertices of $B_S$ are colored with the target coloring.
Before completing the proof of Proposition~\ref{prop:main}, we need one more lemma.

\subsection{Recoloring inside the balls (easy case)}

\begin{lemma}\label{clm:recolball}
Let $k \ge \Delta+1$. Let $\sigma$ and $\eta$ be two $k$-colorings of a graph $G$ which only differ on $X \subseteq V$. Assume that, in each connected component $C$ of $G[X]$, there exists a vertex that  has degree at most $k-2$ or has two neighbors in $V \setminus X$ colored the same. Then there is a recoloring schedule from $\sigma$ to $\eta$ of length at most $k^{O(\diam(X) k)}$.
\end{lemma} 
\begin{proof}
Let $C$ be a connected component of $X$.
For every vertex $v$ of $G[C]$, let $Z_v$ be the set of colors $\sigma$ that appear on neighbors outside $X$, that is on $N(v)\cap(V \setminus X)$.
We assign to every vertex $v$ of $G[C]$ the list of colors $[k] \setminus Z_v$. 
Note that since the total number of colors is $k \ge \Delta+1$, every vertex $x \in C$ has a list of size at least $d_{G[X]}(x)+1$. Moreover, if a vertex $x$ has degree at most $k-2$ in $G$, or two neighbors of $x$ are colored the same in $V \setminus X$, its list has size $d_{G[X]}(x)+2$. 
We claim that we can build a degeneracy ordering of $C$ for which the lists of $C$ are safe, and that consists of $\diam(C) k$ consecutive independent sets. 
Indeed, similarly to earlier in the paper, we can take the vertices of $C$ by layers, corresponding to the distance from $x$, and then split these layers into independent sets using the colors of $\sigma$. 

Finally, by Lemma~\ref{lem:ISrecoloring}, there exists a recoloring  sequence of $G[C]$ from $\sigma$ to $\eta$ which recolors each vertex at most $k^{O(\diam(C) k)}$ times. Since we can treat each connected component of $X$ simultaneously (there is no edge between them), the conclusion follows.
\end{proof}

\subsection{Finishing the proof of Proposition~\ref{prop:main}}

All the previous lemmas can be combined in order to prove Proposition~\ref{prop:main}, that we restate here.

\propMain*

\begin{proof}[Proof of Proposition~\ref{prop:main}]
Let us fix $r=7$. Let $I$ be a maximal independent set at distance $r'=2r+14$. Let $G'=G \setminus B_I$ where $B_I= \cup_{x \in I} B(x,r)$.
By Lemma~\ref{lem:recoloutside}, there is a coloring $\eta'$ which agrees with $\eta$ on $G \setminus B_I$ and a recoloring schedule from $\sigma$ to $\eta'$ of length at most $k^{O(k r)}$. To conclude, we only need to find a recoloring sequence from $\eta'$ to $\eta$, that is to prove that we can recolor all the balls of $B_I$ with their target coloring $\eta$.

For every ball $B_v$ of radius $r$ centered in $v \in I$, we will define a set $B_v'$ which is an extension of $B_v$. We might include some nodes at distance at most $r+5$ from $v$ in order to satisfy the conditions of Lemma~\ref{clm:recolball}. Since $I$ is an independent set at distance $2r+14$, for every $v,w \in I$, the sets $B_v'$ and $B_w'$ will be at distance at least $4$. Let $B_I' = \cup_{v \in I} B_v'$. Since the diameter of each ball $B_v'$ for $v \in S$ is $O(r)$ and all the balls of $B_S'$ are disjoint, we will conclude using Lemma~\ref{clm:recolball}. In the rest of the proof, we restrict to a single ball $B_v$ for $v \in I$ denoted by $B$ for simplicity.

If a vertex of $B$ has two neighbors in $V \setminus B$ colored the same or has degree less than $k-2$, we set $B'=B$. Otherwise, let us prove that by adding a few vertices to $B$ and doing a few recoloring steps, we can apply the Lemma~\ref{clm:recolball}. Note that no vertex of $V \setminus B$ is colored with $k$ in $\eta'$, since it agrees with $\eta$, which is a $k'$-coloring with $k'<k$ by assumption.

Let us consider a path $v_1, v_2, \ldots, v_6$ of vertices such that $v_i$ is at distance $i$ from $B$. 
For every $i \in \{3,4,5\}$, we can remark we can obtain a desired set $B'$ if one of the following holds:
\begin{itemize}
    \item If $\text{deg}(v_i)<\Delta$, then we simply take $B'=B \cup_{j\le i} v_i$ which contains a vertex of degree less than $\Delta$.
    \item If $N(v_i) \setminus v_{i-1}$ is not a clique then let $a,b$ be two neighbors of $v_i$ that are non adjacent. Then, since $d(a,B)$ and $d(b,B)$ are at least two, we can recolor $a$ and $b$ with $k$ in $\eta'$ (the coloring is proper since color $k$ was not used in $\eta$ by assumption). Now, in this new coloring, $B'=B \cup_{j\le i} v_i$ satisfies the condition. 
    (We will recolor $a$ and $b$ to the right color at the very end of the algorithm.)
\end{itemize}

Let us now prove that one of the conditions above must hold.
Assume, for the sake of contradiction, that for every $3 \le i \le 5$, $N(v_i) \setminus v_{i-1}$ is a clique and that all the $v_i$'s have degree at least $\Delta$. 

Let $z$ be a vertex of $N(v_3)$ distinct from $v_2$ and $v_4$ (which exists since $\Delta \ge 3$).
The vertex $z$ is at distance at most $4$ from $B$. Moreover, $v_4z$ is an edge (otherwise $N(v_3) \setminus v_2$ is not a clique). Since $N(v_4) \setminus v_3$ is a clique, $zv_5$ must also be an edge. But then $N(v_5) \setminus v_4$ cannot be a clique: that would mean that $z$ and $v_6$ are adjacent, and then $v_6$ would be at distance 5 from $B$, which is a contradiction.

Now, by Lemma~\ref{clm:recolball}, we can recolor all the vertices of $B'$ with the target coloring~$\eta$ in such a way that every vertex of $B'$ is recolored at most $\Delta^{O(\Delta r)}$ times (since the diameter of $B'$ is at most the diameter of $B$ plus $5$). We then finally recolor, if needed, the two vertices recolored $k$ in the second item of the construction of $B'$ with their real target color in $\eta$.

Since all the balls $B'$ are disjoint and do not share an edge, we can apply these steps in parallel. Moreover, since they are at distance at least $4$, the fact that we recolor a vertex at distance $5$ from $B$ can also be done in parallel. 
This completes the proof of Proposition~\ref{prop:main}.
\end{proof}

\paragraph{Acknowledgments.}

We would like to thank Valentin Gledel for fruitful discussions at the early stage of the project, as well as Quentin Ferro for his helpful study on the specific case of toroidal grids.

\bibliographystyle{abbrv}
\bibliography{biblio}

\end{document}